\documentclass[a4paper,11pt,DIV12]{scrartcl}

\usepackage{kMeans}

\title{Improved Smoothed Analysis of the $k$-Means Method\thanks{An extended
       abstract of this work will appear in \emph{Proc.\ of the 20th ACM-SIAM
       Symposium on Discrete Algorithms (SODA 2009)}.}}

\author{Bodo Manthey\thanks{Work done in part at Yale University, Department of
        Computer Science, supported by the Postdoc-Program of the German
        Academic Exchange Service (DAAD).} \\
        \footnotesize Saarland University \\[\addressminus]
        \footnotesize Department of Computer Science \\[\addressminus]
        \footnotesize \texttt{manthey@cs.uni-sb.de}
   \and
        Heiko R\"oglin\thanks{Supported by a fellowship within the
        Postdoc-Program of the German Academic Exchange Service (DAAD).} \\
        \footnotesize Boston University \\[\addressminus]
        \footnotesize Department of Computer Science \\[\addressminus]
        \footnotesize \texttt{heiko@roeglin.org}}

\begin{document}

\maketitle

\begin{abstract}
The $k$-means method is a widely used clustering algorithm. One of its
distinguished features is its speed in practice. Its worst-case running-time,
however, is exponential, leaving a gap between practical and theoretical
performance. Arthur and Vassilvitskii~\cite{ArthurVassilvitskii:ICP:2006} aimed
at closing this gap, and they proved a bound of $\poly(n^k, \sigma^{-1})$ on the
smoothed running-time of the $k$-means method, where $n$ is the number of data
points and $\sigma$ is the standard deviation of the Gaussian perturbation. This bound,
though better than the worst-case bound, is still much larger than the
running-time observed in practice.

We improve the smoothed analysis of the $k$-means method by showing two upper
bounds on the expected running-time of $k$-means. First, we prove that the
expected running-time is bounded by a polynomial in $n^{\sqrt k}$ and
$\sigma^{-1}$. Second, we prove an upper bound of
$k^{kd} \cdot \poly(n, \sigma^{-1})$, where $d$ is the dimension of the data
space. The polynomial is independent of $k$ and $d$, and we obtain a
polynomial bound for the expected running-time for
$k, d \in O(\sqrt{\log n/\log \log n})$.

Finally, we show that $k$-means runs in smoothed polynomial time for one-di\-men\-sional
instances.
\end{abstract}

\section{Introduction}

The $k$-means method is a very popular algorithm for clustering high-dimensional
data. It is based on ideas by Lloyd~\cite{Lloyd}. It is a local search
algorithm: Initiated with $k$ arbitrary cluster centers, it assigns every data
point to its nearest center, and then readjusts the centers, reassigns
the data points, \ldots\ until it stabilizes. (In Section~\ref{ssec:kmeans}, we
describe the algorithm formally.) The $k$-means method terminates in a local
optimum, which might be far worse than the global optimum.
However, in practice it works very well. It is particularly popular
because of its simplicity and its speed: ``In practice, the number
of iterations is much less than the number of samples'', as Duda et
al.~\cite[Section 10.4.3]{Duda} put it. According to Berkhin~\cite{Berkhin},
the $k$-means method ``is by far the most popular clustering tool used in
scientific and industrial applications.''

The practical performance and popularity of the $k$-means method is at stark
contrast to its performance in theory. The only upper bounds for its
running-time are based on the observation that no clustering appears twice in a
run of $k$-means: Obviously, $n$ points can be distributed among $k$
clusters in only $k^n$ ways. Furthermore, the number of Voronoi partitions of $n$ points in
$\RR^d$ into $k$ classes is bounded by a polynomial in
$n^{kd}$~\cite{InabaEA:WeightedVoronoi:2000}, which yields an upper bound of
$\poly(n^{kd})$. On the other hand,
Arthur and Vassilvitskii~\cite{ArthurVassilvitskii:HowSlow:2006} showed that
$k$-means can run for $2^{\Omega(\sqrt n)}$ iterations in the worst case.

To close the gap between good practical and poor theoretical performance of
algorithms, Spielman and Teng introduced the notion of smoothed
analysis~\cite{SpielmanTeng:SmoothedAnalysisWhy:2004}: An adversary specifies an
instance, and this instance is then subject to slight random perturbations. The
smoothed running-time is the maximum over the adversarial choices of the
expected running-time. On the one hand, this rules out pathological, isolated
worst-case instances. On the other hand, smoothed analysis, unlike average-case
analysis, is not dominated by random instances since the instances are not
completely random; random instances are usually not typical instances and have
special properties with high probability. Thus, smoothed analysis also
circumvents the drawbacks of average-case analysis. For a survey of smoothed
analysis, we refer to Spielman and Teng~\cite{SpielmanTeng:SmoothedFoCM:2006}.

The goal of this paper is to bound the smoothed running-time of the $k$-means
method. There are basically two reasons why the smoothed running-time of the
$k$-means method is a more realistic measure than its worst-case running-time:
First, data obtained from measurements is inherently noisy. So even if the
original data were a bad instance for $k$-means, the data measured is most
likely a slight perturbation of it. Second, if the data possesses a meaningful
$k$-clustering, then slightly perturbing the data should preserve this
clustering. Thus, smoothed analysis might help to obtain a faster $k$-means
method: We take the data measured, perturb it slightly, and then run $k$-means
on the perturbed instance. The bounds for the smoothed running-time carry over
to this variant of the $k$-means method.

\subsection[k-Means Method]{$k$-Means Method}
\label{ssec:kmeans}

An instance for $k$-means clustering is a point set $\points \subseteq \RR^d$
consisting of $n$ points. The aim is to find a clustering
$\cluster 1, \ldots, \cluster k$ of $\points$, i.e., a partition of $\points$,
as well as cluster centers such that the potential
\[
  \sum_{i=1}^k \sum_{x \in \cluster i} \norm{x-c_i}^2
\]
is minimized. Given the cluster centers, every data point should obviously be
assigned to the cluster whose center is closed to it. The name $k$-means stems
from the fact that, given the clusters, the centers $c_1, \ldots, c_k$
should be chosen as the centers of mass, i.e., $c_i =
\frac{1}{|\cluster i|}\sum_{x \in \cluster i} x$ of~$\cluster i$.
The $k$-means method proceeds now as follows:
\begin{enumerate}
\item Select cluster centers $c_1, \ldots, c_k$.
\item \label{next} Assign every $x \in \points$ to the cluster
$\cluster i$ whose cluster center $c_i$ is closest to it.
\item Set $c_i = \frac{1}{|\cluster i|}\sum_{x \in \cluster i} x$.
\item If clusters or centers have changed, goto~\ref{next}. Otherwise,
      terminate.
\end{enumerate}

Since the potential decreases in every step, no clustering occurs twice, and the
algorithm eventually terminates.

\subsection{Related Work}

The problem of finding a good clustering can be approximated arbitrarily
well: B{\u{a}}doiu et al.~\cite{Badoiu}, Matou{\v{s}}ek~\cite{Matousek},
and Kumar et al.~\cite{Kumar} devised polynomial time approximation
schemes with different dependencies on the approximation ratio $(1+\eps)$
as well as $n$, $k$, and $d$: 
$O(2^{O(k \eps^{-2} \log k)} \cdot nd)$, $O(n\eps^{-2k^2d} \log^k n)$,
and $O(\exp(k/\eps) \cdot nd)$,
respectively.

While the polynomial time approximation schemes show that $k$-means clustering
can be approximated arbitrarily well, the method of choice for finding a
$k$-clustering is the $k$-means method due to its performance in practice.
However, the only polynomial bound for $k$-means holds for $d=1$, and
only for instances with polynomial spread~\cite{HarPeled}, which is the
maximum distance of points divided by the minimum distance.

Arthur and Vassilvitskii~\cite{ArthurVassilvitskii:ICP:2006} have analyzed
the running-time of the $k$-means method subject to Gaussian perturbation: The
points are drawn according to independent $d$-dimensional Gaussian distributions
with standard deviation $\sigma$. Arthur and Vassilvitskii proved that the
expected running-time after perturbing the
input with Gaussians with standard deviation $\sigma$ is polynomial in
$n^k$, $d$, the diameter of the perturbed point set, and $1/\sigma$.

Recently, Arthur~\cite{Arthur:Comm:2008} showed that the probability that the
running-time of $k$-means subject to Gaussian perturbations exceeds a polynomial
in $n$, $d$, the diameter of the instance, and $1/\sigma$ is bounded by 
$O(1/n)$. However, his argument does not yield any significant bound
on the expected running-time of $k$-means: The probability of $O(1/n)$ that
the running-time exceeds a polynomial bound is too large to yield an upper bound
for the expected running-time, except for the trivial upper bound of
$\poly(n^{kd})$.

\subsection{New Results}

We improve the smoothed analysis of the $k$-means method by
proving two upper bounds on its running-time.
First, we show that the smoothed running-time of $k$-means
is bounded by a polynomial in $n^{\sqrt k}$ and $1/\sigma$.

\begin{theorem}
\label{thm:sqrtk}
Let $\points \subseteq \RR^d$ be a set of $n$ points drawn according to
independent Gaussian distributions whose means are in $[0,1]^d$. Then the
expected running-time of the $k$-means method on the instance $\points$ is
bounded from above by a polynomial in $n^{\sqrt k}$ and $1/\sigma$.
\end{theorem}

Thus, compared to the previously known bound, we decrease
the exponent by a factor of $\sqrt k$.
Second, we show that the smoothed running-time of $k$-means
is bounded by $k^{kd} \cdot \poly(n, 1/\sigma)$. In particular,
this decouples the exponential part of the bound from the number $n$
of points.

\begin{theorem}
\label{thm:o1}
Let $\points$ be drawn as described in Theorem~\ref{thm:sqrtk}. Then the
expected running-time of the $k$-means method on the instance $\points$ is
bounded from above by $k^{kd} \cdot \poly\bigl(n, 1/\sigma\bigr)$.
\end{theorem}

An immediate consequence of Theorem~\ref{thm:o1} is the following corollary, which
proves that the expected running-time is polynomial in $n$ and $1/\sigma$ if $k$
and $d$ are small compared to $n$. This result is of particular interest
since $d$ and $k$ are usually much smaller than $n$.

\begin{corollary}
\label{cor:poly}
Let $k,d \in O(\sqrt{\log n/\log \log n})$. Let $\points$ be drawn as described
in Theorem~\ref{thm:sqrtk}. Then the expected running-time of $k$-means
on the instance $\points$ is bounded by a polynomial in $n$ and~$1/\sigma$.
\end{corollary}

David Arthur~\cite{Arthur:Comm:2008} presented an insightful proof that
$k$-means runs in time polynomial in $n$, $1/\sigma$, and
the diameter of the instance with a probability of at least $1-O(1/n)$.
It is worth pointing out that his result is orthogonal to our results:
neither do our results imply polynomial running time with probability
$1-O(1/n)$, nor does Arthur's result  
yield any non-trivial bound on the expected running-time
(not even $\poly(n^k, 1/\sigma)$) since the success probability of
$1-O(1/n)$ is way too small. The exception is our result for $d = 1$, which
yields not only a bound on the expectation, but also a bound that holds
with high probability. However, the original definition of smoothed
analysis~\cite{SpielmanTeng:SmoothedAnalysisWhy:2004} is in terms
of expectation, not in terms of bounds that hold with a probability of
$1-o(1)$.

To prove our bounds, we prove a lemma about
perturbed point sets (Lemma~\ref{lemma:NumberOfClosePoints}).
The lemma bounds the number of points close to the boundaries of
Voronoi partitions that arise during the execution of $k$-means.
It might be of independent interest, in particular for smoothed analyses
of geometric algorithms and problems.

Finally, we prove a polynomial bound for the running-time of $k$-means
in one dimension. 

\begin{theorem}
\label{thm:d1}
Let $\points \subseteq \RR$ be drawn according to $1$-dimensional
Gaussian distributions as described in Theorem~\ref{thm:sqrtk}.
Then the expected running-time of $k$-means on $\points$ is polynomial
in $n$ and $1/\sigma$. Furthermore, the probability that the running-time
exceeds a polynomial in $n$ and $1/\sigma$ is bounded by $1/\poly(n)$.
\end{theorem}

We remark that this result for $d = 1$ is not implied by the result
of Har-Peled and Sadri~\cite{HarPeled} that the running-time of one-dimensional
$k$-means is polynomial in $n$ and the spread of the instance. The reason
is that the expected value of the square of the spread is unbounded.

The restriction of the adversarial points to be in $[0,1]^d$
is necessary: Without any bound, the adversary can place the points arbitrarily
far away, thus diminishing the effect of the perturbation.
We can get rid of this restriction and obtain the same results
by allowing the bounds to be polynomial in the diameter of the adversarial
instance.
However, for the sake of clarity and to avoid another parameter, we have chosen
the former model.

\subsection{Outline}

To prove our two main theorems, we first prove a property of perturbed point
sets (Section~\ref{sec:property}):
In any step of the $k$-means algorithm, there are not too many points close
to any of the at most $k^2$ hyperplanes that bisect the centers and that
form the Voronoi regions.
To put it another way: No matter how $k$-means partitions the point set
$\points$ into $k$ Voronoi regions, the number of points close to any boundary
is rather small with overwhelming probability.

We use this lemma in Section~\ref{sec:upper}: First, we use it to prove
Lemma~\ref{lem:upper}, which bounds the expected number of iterations
in terms of the smallest possible distance of two clusters.
Using this bound, we derive a first upper bound for the expected number of
iterations (Lemma~\ref{lem:GeneralDropBy1}), which will result in
Theorem~\ref{thm:o1} later on.

In Sections~\ref{sec:atmostsqrtk} and~\ref{sec:atleastsqrtk},
we distinguish between iterations in which at most $\sqrt k$ or
at least $\sqrt k$ clusters gain or lose points.
This will result in Theorem~\ref{thm:sqrtk}.

We consider the special case of $d=1$ in Section~\ref{sec:1d}. For this
case, we prove an upper bound polynomial in $n$ and $1/\sigma$ until
the potential has dropped by at least $1$.

In Sections~\ref{sec:upper}, \ref{sec:atmostsqrtk}, \ref{sec:atleastsqrtk},
and~\ref{sec:1d} we are only concerned with bounding the
number of iterations until the potential has dropped by at least $1$.
Using these bounds and an upper bound on the potential after the first round, we
will derive Theorems~\ref{thm:sqrtk}, \ref{thm:o1}, and~\ref{thm:d1} as well as
Corollary~\ref{cor:poly} in Section~\ref{sec:putting}.

\subsection{Preliminaries}

In the following, $\points$ is the perturbed instance on which we run
$k$-means, i.e., $\points = \{x_1, \ldots, x_n\} \subseteq \RR^d$ is a set of
$n$ points, where each point $x_i$ is drawn according to a $d$-dimensional
Gaussian distribution with mean $\mu_i \in [0,1]^d$ and standard deviation
$\sigma$.

Inaba et al.~\cite{InabaEA:WeightedVoronoi:2000} proved that the number of
iterations of $k$-means is $\poly\bigl(n^{kd}\bigr)$ in the worst case.
We abbreviate this bound by $W \leq n^{\kappa kd}$ for some constant $\kappa$ in
the following.

Let $D\ge1$ be chosen such that, with a probability of at least $1-W^{-1}$,
every data point from $\points$ lies in the hypercube $\cube:=[-D,1+D]^d$
after the perturbation. In Section~\ref{sec:putting}, we prove that $D$ can be
bounded by a polynomial in $n$ and $\sigma$, and we use this fact in the
following sections. We denote by $\calF$ the \emph{failure event} that there
exists one point in $\points$ that does not lie in the hypercube $\cube$ after
the perturbation. We say that a cluster is \emph{active} in an iteration if it
gains or loses at least one point.

We will always assume in the following that $d \leq n$ and $k \leq n$, and we
will frequently bound both $d$ and $k$ by $n$ to simplify calculations. Of
course, $k \leq n$ holds for every meaningful instance since it does not make
sense to partition $n$ points into more than $n$ clusters. Furthermore, we
can assume $d \leq n$ for two reasons: First, the dimension is usually much
smaller than the number of points, and, second, if $d > n$, then we can project
the points to a lower-dimensional subspace without changing anything.

Let $\clusterSet = \{\cluster 1, \ldots, \cluster k\}$ denote the set of clusters.
For a natural number $k$, let $[k] = \{1,\ldots, k\}$.
In the following, we will assume that number such as $\sqrt k$ are integers. For
the sake of clarity, we do not write down the tedious floor and ceiling
functions that are actually necessary. Since we are only interested in the
asymptotics, this does not affect the validity of the proofs.
Furthermore, we assume in the following sections that $\sigma\le 1$.
This assumption is only made to simplify the arguments and we describe
in Section~\ref{sec:putting} how to get rid of it.

\section{A Property of Perturbed Point Sets}
\label{sec:property}

The following lemma shows that, with high probability,
there are not too many points close to the hyperplanes dividing
the clusters. It is crucial for our bounds for the smoothed running-time:
If not too many points are close to the bisecting hyperplanes, then, eventually,
one point that is further away from the bisecting hyperplanes must go
from one cluster to another, which causes a significant decrease of the potential.

\begin{lemma}
\label{lemma:NumberOfClosePoints}
Let $a \in [k]$ be arbitrary. With a probability of at least
$1-2W^{-1}$, the following holds:
In every step of the $k$-means algorithm (except for the first one) in
which at least $kd/a$ points change their assignment, at least one of
these points has a distance larger than
\[
  \eps :=\frac{\sigma^4}{32n^2dD^2}
           \cdot\left(\frac{\sigma}{3Dn^{3+2\kappa}}\right)^{4a}
\]
from the bisector that it crosses.
\end{lemma}

\begin{proof}
We consider a step of the $k$-means algorithm, and we refer to the
configuration before this step as the \emph{first configuration} and to
the configuration after this step as the \emph{second configuration}. To
be precise, we assume that in the first configuration the positions of
the centers are the centers of mass of the points assigned to them in
this configuration. The step we consider is the reassignment of the
points according to the Voronoi diagram in the first configuration.

Let $B \subseteq \points$ with $|B|=\ell:=kd/a$ be a set of points that change their
assignment during the step. There are at most $n^{\ell}$ choices for the
points in $B$ and at most $k^{2{\ell}}\le n^{2{\ell}}$ choices for the clusters they
are assigned to in the first and the second configuration. We apply a
union bound over all these at most $n^{3\ell}$ choices.

The following sets are defined for all $i,j \in [k]$ and $j \neq i$.
Let $B_i \subseteq B$ be the set of points that leave cluster $\cluster i$.
Let $B_{i,j}\subseteq B_i$ be the set of points assigned to cluster $\cluster i$ in
the first and to cluster $\cluster j$
in the second configuration, i.e., the points in $B_{i,j}$ leave $\cluster i$ and
enter $\cluster j$.
We have $B = \bigcup_i B_i$ and $B_i = \bigcup_{j \neq i} B_{i,j}$.

Let $A_i$ be the set of points that are in $\cluster i$ in the first configuration
except for those in $B_i$.
We assume that the positions of the points in $A_i$ are determined by an adversary.
Since the sets $A_1,\ldots,A_k$ form a partition of the points in
$\points \setminus B$ that has been obtained in the previous step on the basis
of a Voronoi diagram, there are at most $W$ choices for this
partition~\cite{InabaEA:WeightedVoronoi:2000}. We also
apply a union bound over the choices for this partition.

In the first configuration, exactly the points in
$A_i\cup B_i$ are assigned to cluster $\cluster i$.
Let $c_1,\ldots,c_k$ denote the positions of the cluster centers in the first
configuration, that is, $c_i$ is the center
of mass of $A_i\cup B_i$. Since the positions of the points in
$\points \setminus B$ are assumed to be fixed by an adversary, and since we
apply a union bound over the partition
$A_1,\ldots,A_k$, the impact of the set
$A_i$ on position $c_i$ is fixed. However, we want to exploit the
randomness of the points in $B_i$ in the following. Thus, the positions of
the centers are not fixed yet but they depend on the randomness of the points in $B$.
In particular, the bisecting hyperplane $H_{i,j}$ of the clusters $\cluster i$
and $\cluster j$ is not fixed but depends on $B_i$ and $B_j$. 

In order to complete the proof, we have to estimate the
probability of the event
\renewcommand{\theequation}{$\calE$}
\begin{equation}\label{eqn:ProbAllPointsClose}
   \forall i,j\colon \forall b\in B_{i,j}\colon
       \dist(b, H_{i,j}) \leq \eps\:,
\end{equation}%
\addtocounter{equation}{-1}%
where $\dist(x, H) = \min_{y \in H} \norm{x-y}$ denotes the shortest
distance of a point $x$ to a hyperplane $H$. In the following, we
denote this event by $\calE$.
If the hyperplanes $H_{i,j}$ were fixed, the 
probability of $\calE$ could readily be seen to be at most
$\bigl(\frac{2\eps}{\sigma \sqrt{2\pi}}\bigr)^\ell \leq \bigl(\frac\eps\sigma\bigr)^{\ell}$. But the
hyperplanes are not fixed since their positions and orientations depend on the points in the
sets $B_{i,j}$. Therefore, we are only able to prove the following weaker bound in
Lemma~\ref{lemma:ProbAllPointsClose}:
\[
  \Pr{\calE\wedge\neg\calF} \le 
  \left(\frac{3D}{\sigma}\right)^{kd}\cdot
  \left(\frac{32n^2dD^2\eps}{\sigma^4}\right)^{\ell/4}\:,
\]
where $\neg\calF$ denotes the event that, after the perturbation, all
points of $\points$ lie in the hypercube $\cube = [-D, D+1]^d$.
Now the union bound yields the following upper bound on the probability that a
set $B$ with the stated properties exists:
\begin{align*}
  \Pr{\calE} & \le \Pr{\calE\wedge\neg\calF}+\Pr{\calF} \\
  & \le n^{3\ell} W \cdot \left(\frac{3D}{\sigma}\right)^{kd}\cdot
  \left(\frac{32n^2dD^2\eps}{\sigma^4}\right)^{\ell/4}+W^{-1}\\
& =
  n^{3\ell} W \cdot \left(\frac{1}{n^{3+2\kappa}}\right)^{kd}+W^{-1} \\
& \le n^{3\ell+\kappa kd} \cdot
\left(\frac{1}{n^{3+2\kappa}}\right)^{kd}+W^{-1} \\ 
& \leq n^{-\kappa kd} +W^{-1}
  \: \leq \: 2W^{-1}\:.
\end{align*}
The equation is by our choice of $\eps$, the inequalities are due to
some simplifications and $W \leq n^{\kappa kd}$.
\end{proof}

\begin{lemma}
\label{lemma:ProbAllPointsClose}
 The probability of the event $\calE\wedge\neg\calF$ is bounded from
 above by
 \[
\left(\frac{3D}{\sigma}\right)^{kd}\cdot
  \left(\frac{32n^2dD^2\eps}{\sigma^4}\right)^{\ell/4}\:\enspace.
\]
\end{lemma}

\begin{proof}
Let $g_i$ be the random vector that equals the sum of the points
in $B_i$, i.e.,
\[
  g_i := \sum_{b\in B_i}b\enspace.
\]
Due to the union bound, the influence of $A_i$ and $A_j$ on the
hyperplane $H_{i,j}$ is fixed. Since the union bound also
fixes the number of points in $B_i$ and $B_j$, it suffices to know the
sums $g_i$ and $g_j$ to deduce the exact position of the
hyperplane $H_{i,j}$. Hence, once all sums $g_i$ are fixed, all
hyperplanes are fixed as well. The drawback is, of course, that
fixing the sums $g_i$ has an impact on the distribution of the random
positions of the points in $B_i$. We circumvent this problem as follows:
We basically show that if $B_i$ contains $m$ points and the sum
$g_i$ is fixed, then we can still use the randomness of $m-1$ of
these points. For sets $B_i$ that contain at least two points, this
means that we can use the randomness of at least half of its points.
Complications are only caused by sets $B_i$ that consist of a single
point. For such sets, fixing $g_i$ is equivalent to fixing the position of
the point, and we give a more direct analysis without fixing $g_i$.

For $y_i,y_j\in\RR^d$, we denote by $H_{i,j}(y_i,y_j)$ the bisector of
the clusters $\cluster i$ and $\cluster j$ that is obtained for $g_i=y_i$ and
$g_j=y_j$.
Let $\kactive$ be the number of clusters $\cluster i$ with $|B_i|>0$.
Without loss of generality, these are the clusters $\cluster
1,\ldots,\cluster{\kactive}$.
This convention allows
us to rewrite the probability of $\calE\wedge\neg\calF$ as
\begin{multline*}
  \Pr{\forall i,j\colon \forall b\in B_{i,j}\colon
       \dist(b,H_{i,j})\leq \eps\wedge\neg\calF}
 \le \int_{\cube}\cdots\int_{\cube} 
     \left(\prod_{i=1}^{\kactive} f_{g_i}(y_i)\right)\\
 \cdot \condPr{\forall i,j\colon \forall b\in B_{i,j}\colon
       \dist(b, H_{i,j}(y_i,y_j)) \leq \eps}
       {\forall i\colon g_i=y_i}\, \dint{y_{\kactive}}\ldots \dint{y_1}\:,
\end{multline*}
where $f_{g_i}$ is the density of the random vector $g_i$.
We admit that our notation is a bit sloppy: If $|B_{i,j}|>0$ and
$j\notin\{1,\ldots,\kactive\}$, then $H_{i,j}$ depends only on $y_i$.
In this case, we should actually write $H_{i,j}(y_i)$ instead
of $H_{i,j}(y_i,y_j)$ in the formula above. In order to keep the
notation less cumbersome, we ignore this subtlety and assume that 
$H_{i,j}(y_i,y_{j_i})$ is implicitly replaced by $H_{i,j}(y_i)$ whenever
necessary.
Points from different sets $B_i$ and $B_j$ are independent even under
the assumption that the sums $g_i$ and $g_j$ are fixed. Hence, we can
further rewrite the probability as
\begin{multline}
 \label{eqn:IntegrationIndependent}
  \int_{\cube}\cdots\int_{\cube}
  \left(\prod_{i=1}^{\kactive} f_{g_i}(y_i)\right) \\
 \cdot \left(\prod_{i=1}^{\kactive}
     \condPr{\forall j\colon \forall b\in B_{i,j}\colon
       \dist(b, H_{i,j}(y_i,y_j)) \leq \eps}
       {g_i=y_i}\right)\,\dint{y_{\kactive}}\ldots \dint{y_1}\:.
\end{multline}
Now let us consider the probability
\begin{equation}\label{eqn:ProbFixedI}
  \condPr{\forall j\colon \forall b\in B_{i,j}\colon
       \dist(b, H_{i,j}(y_i,y_j)) \leq \eps}
       {g_i=y_i}
\end{equation}
for a fixed $i$ and for fixed values $y_i$ and $y_j$. To simplify the
notation, let $B_i=\{b_1,\ldots,b_m\}$, and let the
corresponding hyperplanes (which are fixed because $y_i$ and the $y_j$'s
are given) be $H_1,\ldots,H_m$. (A hyperplane may occur several times
in this list if more than one point goes from $\cluster i$ to some
cluster $\cluster j$.) Then the probability simplifies to
\[
  \condPr{\forall j\colon \dist(b_j, H_{j}) \leq \eps}
  {g_i=y_i}\enspace.
\]
We distinguish between the cases $m=1$ and $m>1$.

\paragraph{Case 1: $m=1$.} The probability degenerates to
\begin{equation}
\label{eqn:Indicator}
    \condPr{\dist(b_1, H_{1}) \leq \eps}{g_i=b_1=y_i}
  = \begin{cases}
      1 & \text{if $y_i$ is $\eps$-close to $H_1$,} \\
      0 & \text{otherwise.}     
    \end{cases}
\end{equation}
So, given $g_i = y_i$, there is no randomness in the event that
$y_i$ is $\eps$-close to $H_1$.
Choose $j_i$ such that $b_1\in B_{i,j_i}$ and denote by $\indi_i(y_i,y_{j_i})$ the
indicator variable defined in~\eqref{eqn:Indicator}. We replace the corresponding
probability in~\eqref{eqn:IntegrationIndependent} by $\indi_i(y_i,y_{j_i})$.

\paragraph{Case 2: $m>1$.}
Let $H_j(\eps)$ be the slab of width $2\eps$ around $H_j$, i.e.,
$H_j(\eps) = \{x \in \RR^d \mid \dist(x,H_j) \leq \eps\}$.
Let $f$ be the joint density of the
random vectors $b_1,\ldots,b_{m-1},g_i$. Then the
probability~\eqref{eqn:ProbFixedI} can be bounded from above by
\[
  \int_{z_1\in H_1(\eps)}\cdots \int_{z_{m-1}\in H_{m-1}(\eps)}
  \frac{f(z_1,\ldots,z_{m-1},y_i)}{f_{g_i}(y_i)}\, \dint{z_{m-1}}\ldots
  \dint{z_1}\enspace.
\]
Now let $f_i$ be the density of the random vector $b_i$. This allows us
to rewrite the joint density, and we obtain the upper bound
\begin{align*}
  &\quad\int_{z_1\in H_1(\eps)}\cdots \int_{z_{m-1}\in H_{m-1}(\eps)}
  \frac{f_1(z_1)\cdot\ldots\cdot f_{m-1}(z_{m-1})\cdot
  f_{m}(y_i-\sum_{j=1}^{m-1}z_j)} {f_{g_i}(y_i)}\, \dint{z_{m-1}}\ldots
  \dint{z_1}\\
  &\le \frac{1}{f_{g_i}(y_i)\cdot\sigma^d} 
  \int_{z_1\in H_1(\eps)}\cdots \int_{z_{m-1}\in H_{m-1}(\eps)}
  f_1(z_1)\cdot\ldots\cdot f_{m-1}(z_{m-1})\, \dint{z_{m-1}}\ldots
  \dint{z_1}\\
  &= \frac{1}{f_{g_i}(y_i)\cdot\sigma^d}
  \left(\prod_{i=1}^{m-1}
    \int_{z_i\in H_i(\eps)}f_i(z_i)\,\dint{z_i}
  \right)\\  
  & \le \frac{1}{f_{g_i}(y_i)\cdot\sigma^d}
  \cdot\left(\frac{\eps}{\sigma}\right)^{m-1}\enspace.
\end{align*}
The first inequality follows from $f_m(\cdot) \leq \sigma^{-d}$, and the
last inequality follows because the probability that a
Gaussian random vector assumes a position within distance $\eps$
of a given hyperplane is at most $\eps/\sigma$.

Now we plug the bounds derived in Cases~1 and~2
into~\eqref{eqn:IntegrationIndependent}. Let $\kone$ be the
number of clusters $\cluster i$ with $|B_i|=1$, and let
us assume that these are the clusters $\cluster 1,\ldots,\cluster{\kone}$.
Let
\[
   \ktwo = |\{i\mid |B_i|>1\}|
   \text{\quad and\quad}
   m'  = \sum_{i,|B_i|>1}(|B_i|-1)\:,
\]
that is, $\ktwo$ is the number of clusters that lose more than
one point, and $m'$ is the number of points leaving those clusters
minus $\ktwo$, i.e., the number of points whose randomness we exploit.
Note that $\kactive = \ktwo + \kone$.
Then~\eqref{eqn:IntegrationIndependent}
can be bounded from above by
\begin{equation}\label{eqn:RemainingProbability}
  \frac{1}{\sigma^{\ktwo d}}
  \cdot\left(\frac{\eps}{\sigma}\right)^{m'}
  \int_{\cube}\cdots\int_{\cube}
  \left(\prod_{i=1}^{\kone}
     f_{g_i}(y_i)\cdot\indi_i(y_i,y_{j_i})\right)\,\dint{y_{\kactive}}\ldots
     \dint{y_1}\:,
\end{equation}
where $f_{g_i}$ cancels out for $i > \kone$.
Observe that for fixed $y_{j_i}$ the term
\begin{equation}\label{eqn:IsolatedProb}
  \int_{\cube}f_{g_i}(y_i)\cdot\indi_i(y_i,y_{j_i})\,\dint{y_i}
\end{equation}
describes the probability that the point $b$ with $B_i=\{b\}$ lies in
$\cube$ and is at distance at most $\eps$ from the bisector of $\cluster
i$ and $\cluster {j_i}$. For $y_{j_i}\in\cube$, the point $b$ can only
lie in the hypercube $\cube$ if it has a distance of at most
$\sqrt{d}(1+2D)\le 3\sqrt{d}D$ from $y_{j_i}$. Hence, we can use 
Lemma~\ref{lemma:SinglePointDeterminesCenter} with $\delta = 3\sqrt d D$
to obtain that the probability in~\eqref{eqn:IsolatedProb} is upper bounded by
\[
  \frac{2\sqrt{n3\sqrt{d}D\eps}}{\sigma}
  \le\frac{4\sqrt{n\sqrt{d}D\eps}}{\sigma}\:.
\]
Since there can be circular dependencies like, e.g., $j_i = i'$ and
$j_{i'} = i$, it might not be possible to reorder the integrals
in~\eqref{eqn:RemainingProbability} such that all terms become isolated
as in~\eqref{eqn:IsolatedProb}. We resolve these dependencies by only
considering a subset of the clusters. To make this more precise, consider
a graph whose nodes are the clusters and that has a directed edge
from $\cluster i$ to $\cluster j$ if $|B_i|=|B_{i,j}|=1$, i.e., for every $i$ with
$|B_i|=1$, there is an edge from $\cluster i$ to $\cluster{j_i}$. This graph
contains exactly $\kone$ edges and we can identify a subset of
$k' \geq \kone/2$ edges that is cycle-free. The subset $\clusterSet'$ of
clusters that we consider consists of the tails of these edges. Since every node in
the graph has an out-degree of at most one, $\clusterSet'$ consists of exactly $k'$
clusters.
For each cluster $\cluster i$ not contained in $\clusterSet'$, we replace the
corresponding $\indi_i(y_i,y_{j_i})$ by the trivial upper bound of $1$.
Without loss of generality, the identified subset $\clusterSet'$ consists of the clusters
$\cluster 1,\ldots,\cluster{k'}$ and it is topologically sorted in the sense that
$i<j_i$ for all $i\in\{1,\ldots,k'\}$. Given this,
\eqref{eqn:RemainingProbability} can be bounded from above by
\begin{equation*}
  \frac{1}{\sigma^{\ktwo d}}
  \cdot\left(\frac{\eps}{\sigma}\right)^{m'}
  \underbrace{\int_{\cube}
  \cdots\int_{\cube}}_{\kactive - k' \text{ integrals}}
  \underbrace{\int_{\cube}f_{g_{k'}}(y_{k'})\cdot\indi_{k'}(y_{k'},y_{j_{k'}})
  \cdots\int_{\cube}}_{k' \text{ integrals}}
  f_{g_1}(y_1)\cdot\indi_1(y_1,y_{j_1})
  \,\dint{y_1}\ldots \dint{y_{\kactive}}\enspace.
\end{equation*}
Evaluating this formula from right to left according
to Lemma~\ref{lemma:SinglePointDeterminesCenter} yields
\[
  (3D)^{dk} \cdot \frac{1}{\sigma^{\ktwo d}}
  \cdot\left(\frac{\eps}{\sigma}\right)^{m'}
  \left(\frac{4\sqrt{n\sqrt{d}D\eps}}{\sigma}\right)^{k'}
  = (3D)^{dk} \cdot \left(4\sqrt{n\sqrt{d}D}\right)^{k'}
  \frac{\eps^{m'+k'/2}}{\sigma^{\ktwo d+m'+k'}}
  \:,
\]
where the term $(3D)^{dk}$ comes from the $\kactive - k' \leq k$ integrals
over $y_{k'+1},\ldots, y_{\kactive}$: Each of these integrals is over
the hypercube $\cube$, which has a volume of $(2D+1)^d\le(3D)^d$.
The definitions directly yield that $\ktwo \le k$ and
$m'+k'\leq m' + \kone \leq \ell$. Furthermore,
\[
   m'\ge\frac{\sum_{i,|B_i|>1}|B_i|}{2}
   \text{\quad and\quad}
   k'\ge \frac{\kone}2 = \frac{\sum_{i,|B_i|=1}|B_i|}{2}
\]
implies $m'+k'/2\ge\ell/4$. Altogether, this yields the desired upper
bound of
\[
 (3D)^{dk} \left(4\sqrt{n\sqrt{d}D}\right)^{\ell}\cdot
  \frac{\eps^{\ell/4}}{\sigma^{kd+\ell}}
  = \left(\frac{3D}{\sigma}\right)^{kd}\cdot
  \left(\frac{32n^2dD^2\eps}{\sigma^4}\right)^{\ell/4}
\]
for the probability of the event $\calE \wedge \neg\calF$.
\end{proof}

\begin{lemma}\label{lemma:SinglePointDeterminesCenter}
Let $o\in\RR^d$ and $p\in\RR^d$ be arbitrary points and
let $r$ denote a random point chosen according to a $d$-dimensional
normal distribution with arbitrary mean and standard deviation~$\sigma$.
Moreover, let $\ell\in\{0,\ldots,n-1\}$, and let
\[
  q = \frac{\ell}{\ell+1} \cdot p + \frac{1}{\ell+1} \cdot r
\]
be a convex combination of $p$ and $r$. Then the probability that $r$ is
\begin{enumerate}[(i)]
\item at a distance of at most $\delta>0$ from $o$ and
\item at a distance of at most $\eps>0$ from the bisector of $o$ and $q$
\end{enumerate}
is bounded from above by
\[
  \frac{2\sqrt{n\delta\eps}}{\sigma}\enspace. 
\]
\end{lemma}
\begin{proof}
For ease of notation, we assume that $o$ is the origin of the
coordinate system, i.e., $o=(0, \ldots, 0)$. Due to rotational symmetry, we can
also assume that $p = (0,p_2,0,\ldots,0)$ for some $p_2 \in\RR$.
Let $r=(r_1,\ldots,r_d)$, and assume that the
coordinates $r_2,\ldots,r_d$ are fixed arbitrarily. We only exploit the
randomness of the coordinate $r_1$, which is a one-dimensional Gaussian
random variable with standard deviation $\sigma$. The condition that $r$
has a distance of at most $\eps$ from the bisector of $o$ and $q$ 
can be expressed algebraically as
\[
  \frac{q-o}{\norm{q-o}}\cdot\left(\frac{q+o}{2}-r\right) 
  \in\left[-\eps,\eps\right]\:.
\]
Since $o = (0,\ldots, 0)$, this simplifies to
\[
  \frac{q}{\norm{q}}
  \cdot\left(\frac{q}{2}-r\right)\in\left[-\eps,\eps\right]
  \:\iff\: q\cdot\left(\frac{q}{2}-r\right)
  \in\bigl[-\norm{q}\eps, \norm{q}\eps\bigr]\:.
\]
Since $r_i$ is fixed for $i\neq1$, also the coordinates $q_i$ of $q$ are fixed
for $i\neq1$. Setting $2\lambda=1/(\ell+1)$ and
exploiting that the first coordinate of $p$ is $0$, we
can further rewrite the previous expression as
\[
  \begin{pmatrix}
    2\lambda r_1 \\
    q_2 \\
    \vdots \\
    q_d 
  \end{pmatrix}
  \cdot
  \begin{pmatrix}
  (\lambda-1) r_1\\
  q_2/2-r_2 \\
  \vdots\\
  q_d/2-r_d
  \end{pmatrix}
  \in\bigl[-\norm{q}\eps,\norm{q}\eps\bigr]
  \:,
\]
which is equivalent to
\[
  2\lambda(\lambda-1) r_1^2
  +\sum_{i=2}^d q_i(q_i/2-r_i)
  \in\bigl[-\norm{q}\eps,\norm{q}\eps\bigr]\:.
\]
Since the coordinates $q_i$ and $r_i$ are fixed for $i\neq 1$, this
implies that $r$ can only be at distance $\eps$ from the bisector of $p$ and $q$ if $r_1^2$
falls into a fixed interval of length
\[
  \frac{2\norm{q}\eps}{2\lambda(1-\lambda)}
  = \frac{2(\ell+1)\norm{q}\eps}{1-\frac{1}{2(\ell+1)}}
  \le 4n\norm{q}\eps \: .
\]
As we only consider the event that $q$ has a distance of at most $\delta$ from $o$,
we can replace $\norm{q}$ by $\delta$ in the expression above, leaving us
with the problem to find an upper bound for the probability that the random variable
$r_1^2$ assumes a value in a fixed interval of length at most $4n\delta\eps$. For this
to happen, $r_1$ has to assume a value in one of two intervals, each of length at most
$2\sqrt{n\delta\eps}$. Now $r_1$ follows a Gaussian distribution with standard
deviation $\sigma$. This means that the corresponding density is bounded from above
by $(\sqrt{2 \pi} \sigma)^{-1}$. Thus, the probability of this event is
at most
\[
   \frac{4\sqrt{n\delta\eps}}{\sqrt{2\pi}\sigma}
   < \frac{2\sqrt{n\delta\eps}}{\sigma}\:. \qedhere
\]
\end{proof}

\section{An Upper Bound}
\label{sec:upper}

Lemma~\ref{lemma:NumberOfClosePoints} yields an upper bound on the number of
iterations that $k$-means needs: Since there are only few points close to
hyperplanes, eventually a point switches from one cluster to another that
initially was not close to a hyperplane. The results of this section lead to the
proof of Theorem~\ref{thm:o1} in Section~\ref{ssec:o1}.

First, we bound the number of iterations in terms of the distance $\Delta$ of
the closest cluster centers that occur during the run of $k$-means.

\begin{lemma}
\label{lem:upper}
For every $a\in[k]$, with a probability of at least $1-3W^{-1}$, every sequence
of $k^{kd/a}+1$ consecutive steps of the $k$-means algorithm (not including the
first one) reduces the potential by at least 
\[
  \frac{\eps^2\cdot \min\{\Delta^2,1\}}{36dD^2k^{kd/a}}\:,
\]
where $\Delta$
denotes the smallest distance of two cluster centers that occurs during the
sequence and $\eps$ is defined as in
Lemma~\ref{lemma:NumberOfClosePoints}.
\end{lemma}

\begin{proof}
Consider the configuration directly before the sequence of steps is
performed. Due to Lemma~\ref{lemma:NumberOfClosePoints}, the probability
that more than $kd/a$ points are within distance $\eps$ of one of the
bisectors is at most $2W^{-1}$. Additionally, only with a
probability of at most $W^{-1}$ there exists a point from $\points$ that does
not lie in the hypercube $\cube$. Let us assume in the following
that none of these failure events occurs, which is the case with
a probability of at least $1-3W^{-1}$.

These points can assume
at most $k^{kd/a}$ different configurations. Thus, during the considered sequence, at least
one point that is initially not within distance $\eps$ of one of the
bisectors must change its assignment. Let us call this point $x$, and
let us assume that it changes from cluster $\cluster 1$ to cluster $\cluster 2$.
Furthermore, let $\delta$ be the distance of the centers of $\cluster 1$
and $\cluster 2$ before the sequence, and let $c_1$ and $c_2$ be the
positions of the centers before the sequence. We distinguish two 
cases. First, if $x$ is closer to $c_2$ than to $c_1$ already in the
beginning of the sequence, then
$x$ will change its assignment in the first step. Let $v = \frac{c_2 - c_1}{\|c_2 - c_1\|}$
be the unit vector in $c_2 - c_1$ direction. We have $c_2 - c_1 = \delta v$ and
$(2x - c_1 - c_2) \cdot v = \alpha v$ for some $\alpha \geq 2\eps$. Then the switch of
$x$ from $\cluster 1$ to $\cluster 2$ reduces the potential by at least
\[
       \norm{x-c_1}^2-\norm{x-c_2}^2
   =   (2x-c_1-c_2) \cdot (c_2-c_1)
  \geq 2 \eps\delta \geq 2 \eps \Delta
  \ge \frac{\eps^2\cdot \min\{\Delta^2,1\}}{4D^2dk^{kd/a}}\:,
\]
where the last inequality follows from $\eps\le 1$ and $D\ge1$.
This completes the first case. Second, if $x$
is closer to $c_1$ than to $c_2$, then
\[
   \norm{x-c_2}^2-\norm{x-c_1}^2 \ge 2\eps\delta \:,
\]
and hence,
\[
   \norm{x-c_2}-\norm{x-c_1} \ge
   \frac{2\eps\delta}{\norm{x-c_2}+\norm{x-c_1}}
   \ge \frac{\eps\delta}{36D \sqrt{d}}\:.
\]
In this case, $x$ can only change to cluster $\cluster 2$ after at least
one of the centers of $\cluster 1$ or $\cluster 2$ has moved. Consider the centers
of \cluster 1 and \cluster 2 immediately before the reassignment of $x$
from $\cluster 1$ to $\cluster 2$. Let $c_1'$ and
$c_2'$ denote these centers. Then
\[
   \norm{x-c_1'}-\norm{x-c_2'} > 0 \:.
\]
By combining these observations with the triangle inequality, we obtain
\begin{align*}
   \norm{c_1-c_1'}+\norm{c_2-c_2'}
   & \ge \bigl(\norm{x-c_1'} - \norm{x-c_1}\bigr)
          + \bigl(\norm{x-c_2} - \norm{x-c_2'}\bigr)\\
   & = \bigl(\norm{x-c_1'} - \norm{x-c_2'}\bigr)
          + \bigl(\norm{x-c_2} - \norm{x-c_1} \bigr)
   \ge \frac{\eps\delta}{\sqrt{d}3D}\:.
\end{align*}
This implies that one of the centers must
have moved by at least $\eps\delta/(6\sqrt{d}D)$ during the considered
sequence of steps. Each time the center moves by some amount $\xi$, the potential
drops by at least $\xi^2$ (see Lemma~\ref{lemma:MovementImprovement}).
Since this function is concave, the smallest potential drop is obtained if the center moves by
$\eps\delta/(6\sqrt{d}Dk^{kd/a})$ in each iteration. Then the decrease of the
potential due to the movement of the center is at least
\[
   k^{kd/a}\cdot \left(\frac{\eps\delta}{6\sqrt{d}Dk^{kd/a}}\right)^2
   \ge \frac{\eps^2\Delta^2}{36dD^2k^{kd/a}}\:,
\]
which concludes the proof.
\end{proof}

In order to obtain a bound on the number of iterations that $k$-means
needs, we need to bound the distance $\Delta$ of the closest cluster
centers. This is done in the following lemma, which exploits
Lemma~\ref{lem:upper}.
The following lemma is the crucial ingredient of the proof of
Theorem~\ref{thm:o1}.

\begin{lemma}
\label{lem:GeneralDropBy1}
Let $a \in [k]$ be arbitrary.
Then the expected number of steps until the potential drops by
at least $1$ is bounded from above by
\[
 \gamma\cdot k^{2kd/a}\cdot
  nkd\left(\frac{d^2n^4D}{\sigma\eps}\right)^2
\]
for a sufficiently large absolute constant $\gamma$.
\end{lemma}
\begin{proof}
With a probability of at least $1-3W^{-1}$, the number of iterations
until the potential drops by at least
\[
  \frac{\eps^2\cdot \min\{\Delta^2,1\}}{36dD^2k^{kd/a}}
\]
is at most $k^{kd/a}+1$ due to Lemma~\ref{lem:upper}. We estimate the
contribution of the failure event, which occurs only with probability
$3W^{-1}$, to the expected running time by $3$ and ignore it in the
following. Let $T$ denote the random variable that equals the
number of sequences of length $k^{kd/a}+1$ until the potential has
dropped by one.

The random variable $T$ can only exceed $t$ if
\[
   \min\{\Delta^2,1\} \le \frac{36dD^2k^{kd/a}}{\eps^2\cdot t}\:, 
\]
leading to the following bound on the
expected value of $T$:
\begin{align*}
  \Ex{T} & = \sum_{t=1}^{W}\Pr{T\ge t}
   \le \int_{0}^{W}\PrB{\min\{\Delta^2,1\} \le
   \frac{36dD^2k^{kd/a}}{\eps^2\cdot t}}\,\dint t\\
    & \le t' + 
   \int_{t'}^{W}\PrB{\Delta \le 
   \frac{6\sqrt{d}Dk^{kd/(2a)}}{\eps\cdot \sqrt{t}}}\,\dint t\:,
\end{align*}
for
\[
   t' =
   \left(\frac{(24d+96)n^4\sqrt{d}Dk^{kd/(2a)}}{\sigma\eps}\right)^2\:.
\]

Let us consider a situation reached by $k$-means in which there are two
clusters $\cluster 1$ and $\cluster 2$ whose centers are at a distance
of $\delta$ from each other. We denote the positions of these centers by
$c_1$ and $c_2$. Let $H$ be the bisector between $c_1$ and $c_2$.
The points
$c_1$ and $c_2$ are the centers of mass of the points assigned to
\cluster 1 and \cluster 2, respectively. From this, we can conclude
the following: for every
point that is assigned to $\cluster 1$ or $\cluster 2$ and that has
a distance of at least $\delta$ from the bisector $H$,
as compensation another point must be assigned to $\cluster 1$ or
$\cluster 2$ that has a distance of at most $\delta/2$ from $H$.
Hence, the total number of points assigned to $\cluster 1$ or $\cluster 2$
can be at most twice as large as the total number of points assigned
to $\cluster 1$ or $\cluster 2$ that are at a distance of at
most $\delta$ from $H$. Hence, there can only exist two
centers at a distance of at most $\delta$ if one of the following two
properties is met:
\begin{enumerate}
  \item There exists a hyperplane from which more than $2d$ points have
   a distance of at most $\delta$.
  \item There exist two subsets of points whose union has cardinality
  at most $4d$ and whose centers of mass are at a distance of at most $\delta$.
\end{enumerate}
The probability that one of these events occurs can be bounded as follows
using a union bound and Lemma~\ref{lem:sepaprob} (see also Arthur and
Vassilvitskii~\cite[Proposition~5.6]{ArthurVassilvitskii:ICP:2006}):
\[
  \Pr{\Delta\le\delta} \le 
   n^{2d}\left(\frac{4d\delta}{\sigma}\right)^{2d-d}
   + (2n)^{4d}\cdot\left(\frac{\delta}{\sigma}\right)^d
   \le \left(\frac{(4d+16)n^4\delta}{\sigma}\right)^{d}\enspace.
\]
Hence,
\[
   \PrB{\Delta \le 
   \frac{6\sqrt{d}Dk^{kd/(2a)}}{\eps\cdot \sqrt{t}}}
   \le \left(\frac{\sqrt{t'}}{\sqrt{t}}\right)^{d}
\]
and, for $d\ge3$, we obtain
\begin{align*}
  \Ex{T} & \le t'+
   \int_{t'}^{W}\left(\frac{\sqrt{t'}}{\sqrt{t}}\right)^{d} \dint t \\
   & \le
   t'+t'^{d/2}\left[\frac{1}{(-d/2+1)\cdot
   t^{d/2-1}}\right]_{t'}^{\infty} = \frac{d}{d-2}\cdot t' \le 2\kappa nkd\cdot t'\:.
\end{align*}
For $d=2$, we obtain
\[
  \Ex{T}  \le t'+
   \int_{t'}^{W}\left(\frac{\sqrt{t'}}{\sqrt{t}}\right)^{d} \dint t
    \le t'+t'\cdot\bigl[\ln(t)\bigr]_{1}^{W}
   = t'\cdot\left(1+\ln(W)\right)
   \le 2\kappa nkd\cdot t'\:.
\]

Altogether, this shows that the expected number of steps until the
potential drops by at least $1$ can be bounded from above by
\[
 2+\left(k^{kd/a}+1\right)\cdot
 2\kappa nkd \cdot
 \left(\frac{(24d+96)n^4\sqrt{d}Dk^{kd/(2a)}}{\sigma\eps}\right)^2\:,
\]
which can, for a sufficiently large absolute constant $\gamma$, be
bounded from above by
\[
  \gamma\cdot k^{2kd/a}\cdot
  nkd \cdot \left(\frac{d^2n^4D}{\sigma\eps}\right)^2\:.\qedhere
\] 
\end{proof}

\section[Iterations with at most sqrt(k) Active Clusters]%
        {Iterations with at most $\sqrt{k}$ Active Clusters}
\label{sec:atmostsqrtk}

In this and the following section, we aim at proving the main lemmas that lead
to Theorem~\ref{thm:sqrtk}, which we will prove in Section~\ref{ssec:sqrtk}.
To do this, we distinguish two cases: In this section, we deal with the case
that at most $\sqrt k$ are active. In this case, either few points change
clusters, which yields a potential drop caused by the movement of the centers.
Or many points change clusters. Then, in particular, many points switch between
two clusters, and not all of them can be close to the hyperplane bisecting
the corresponding centers, which yields the potential drop in this case.

We define an \emph{epoch} to be a sequence of consecutive iterations in which no
cluster center assumes more than two different positions.
Equivalently, there are at most two different sets $\cluster{i}',
\cluster{i}''$ that every cluster $\cluster i$ assumes.
The obvious upper bound for the length of an epoch is $2^k$, which is stated
also by Arthur and Vassilvitskii~\cite{ArthurVassilvitskii:ICP:2006}:
After that many iterations, at least one cluster must have assumed a
third position. For our analysis, however, $2^k$ is too big, and we
bring it down to a constant.

\begin{lemma}
\label{lem:epochelength}
The length of any epoch is less than four.
\end{lemma}

\begin{proof}
Let $x$ be any data point that changes from one cluster to another during
an epoch, and let $i_1, i_2, \ldots, i_\ell$ be the indices of the
different clusters to which $x$ belongs in that order. (We have $i_j \neq
i_{j+1}$, but $x$ can change back to a cluster it has already visited.
So, e.g., $i_j = i_{j+2}$ is allowed.) For every $i_j$, we then have two
different sets $\clp j$ and $\clpp j$ with centers $\ccp j$ and $\ccpp j$
such that $x \in \clpp j \setminus \clp j$.
Since $x$ belongs always to at exactly one cluster, we have $\cl j = \clp j$ for
all except for one $j$ for which $\cl j = \clpp j$.
Now assume that $\ell \geq 4$. Then, when changing from $\cl 1$ to $\cl 2$, we have
$\|x-\cc 2'\| < \|x-\cc 4'\|$ since $x$ prefers $\cl 2$ over $\cl 4$ and, when changing
to $\cl 4$, we have $\|x-\cc 4'\| < \|x-\cc 2'\|$. This contradicts the assumption
that $\ell \geq 4$.

Now assume that $x$ does not change from $\cl j$ to
$\cl{j+1}$ for a couple of steps,
i.e., $x$ waits until it eventually changes clusters. Then the reason for eventually
changing to $\cl{j+1}$ can only be that either $\cl{j}$ has changed to some
$\clt j$, which makes $x$ prefer $\cl{j+1}$. But, since
$\clt j \neq \clpp j$ and $x \in \clt j$, we have a third cluster for $\cl j$.
Or $\cl{j+1}$ has changed to $\clt{j+1}$, and $x$ prefers $\clt{j+1}$.
But then $\clt{j+1} \neq \clp j$ and $x \notin \clt{j+1}$, and we have
a third cluster for $\cl{j+1}$.

We can conclude that $x$ visits at most three different clusters, and changes its
cluster in every iteration of the epoch. Furthermore, 
the order in which $x$ visits its clusters is periodic with
a period length of at most three. Finally, even a period length of three is impossible: Suppose $x$ visits
$\cl 1$, $\cl 2$, and $\cl 3$. Then, to go from $\cl j$ to $\cl{j+1}$ (arithmetic
is modulo $3$), we have $\|x - \cc{j+1}'\| < \|x -
\cc{j-1}'\|$. Since this holds for $j = 1,2,3$, we have a contradiction.

This holds for every data point. Thus, after at most four iterations
either $k$-means terminates, which is fine, or some cluster assumes a third configuration,
which ends the epoch, or some clustering repeats, which is impossible.
\end{proof}

Similar to Arthur and Vassilvitskii~\cite{ArthurVassilvitskii:ICP:2006},
we define a \emph{key-value} to be an expression of the form
$K = \frac st \cdot \mass(S)$, where $s, t \in \NN$, $s \leq n^2$, $t < n$,
and $S \subseteq \points$ is a set of at most $4 d \sqrt k$ points.
(Arthur and Vassilvitskii allow up to $4dk$ points.)
For two key-values $K_1, K_2$, we write $K_1 \equiv K_2$ if and only if
they have identical coefficients for every data point.

We say that $\points$ is \emph{$\delta$-sparse} if, for every key-values
$K_1, K_2, K_3, K_4$ with $\|K_1+K_2-K_3-K_4\| \leq \delta$, we have
$K_1 + K_2 \equiv K_3 + K_4$.

\begin{lemma}
\label{lem:sparseprob}
The probability that the point set $\points$ is not $\delta$-sparse is
at most
\[
  n^{16d\sqrt k+12} \cdot \left(\frac{n^4 \delta}{\sigma}\right)^d.
\]
\end{lemma}

\begin{proof}
Let us first bound the number of possible key-values:
There are at most $n^3$ possibilities for choosing $s$ and $t$
and $n^{4 d \sqrt k}$ possibilities for choosing the set $S$.
Thus, there are at most $n^{16d\sqrt k + 12}$ possibilities
for choosing four key-values $K_1, \ldots, K_4$. We fix
$K_1, \ldots, K_4$ arbitrarily. The rest follows from
a union bound and the proof of Proposition~5.3 of
Arthur and Vassilvitskii~\cite{ArthurVassilvitskii:ICP:2006}.
\end{proof}

After four iterations, one cluster has assumed a third center or $k$-means
terminates. This yields the following lemma (see also Arthur
and Vassilvitskii~\cite[Corollary~5.2]{ArthurVassilvitskii:ICP:2006}).

\begin{lemma}
\label{lem:deltasparsedrop}
Assume that $\points$ is $\delta$-sparse. Then, in every sequence of
four consecutive iterations that do not lead to termination and such
that in every of these iterations
\begin{itemize}
\item at most $\sqrt k$ clusters are active and
\item each cluster gains or loses at most $2d\sqrt k$ points,
\end{itemize}
the potential decreases by at least
$\frac{\delta^2}{4n^4}$.
\end{lemma}

We say that $\points$ is $\eps$-separated if, for every hyperplane $H$,
there are at most $2d$ points in $\points$ that are within
distance $\eps$ of $H$. The following lemma, due to Arthur and
Vassilvitskii~\cite[Proposition 5.6]{ArthurVassilvitskii:ICP:2006}, shows that
$\points$ is likely to be $\eps$-separated.

\begin{lemma}[Arthur and Vassilvitskii~\cite{ArthurVassilvitskii:ICP:2006}]
\label{lem:sepaprob}
The point set $\points$ is not $\eps$-separated with a
probability of at most
\[
  n^{2d}\cdot \left(\frac{4d\eps}{\sigma}\right)^d.
\]
\end{lemma}

Given that $\points$ is $\eps$-separated, every iteration with at most $\sqrt k$
active clusters in which one cluster gains or loses at least $2d \sqrt
k$ points yields a significant decrease of the potential.

\begin{lemma}
\label{lem:epssepdrop}
Assume that $\points$ is $\eps$-separated. For every iteration with at most
$\sqrt k$ active clusters, the following holds: If a cluster gains or loses more
than $2d \sqrt k$ points, then the potential drops by at least $2\eps^2/n$.
\end{lemma}

This lemma is similar to Proposition 5.4 of Arthur and
Vassilvitskii~\cite{ArthurVassilvitskii:ICP:2006}. We present here a corrected
proof based on private communication with Vassilvitskii.

\begin{proof}
If a cluster $\cluster i$ gains or loses more than $2d \sqrt k$ points in a
single iteration with at most $\sqrt k$ active clusters, then there exists
another cluster $\cluster j$ with which $C_i$ exchanges at least $2d+1$ points.
Since $\points$ is $\eps$-separated, one of these points, say, $x$, must be at a
distance of at least $\eps$ from the hyperplane bisecting the cluster centers
$c_i$ and $c_j$. Assume that
$x$ switches from $\cluster i$ to $\cluster j$.

Then the potential decreases by at least $\|c_i-x\|^2 - \|c_j-x\|^2
= (2x - c_i - c_j) \cdot (c_j - c_i)$. Let $v$ be the unit vector
in $c_j - c_i$ direction. Then $(2x - c_i - c_j) \cdot v \geq 2\eps$.
We have $c_j - c_i = \alpha v$ for $\alpha = \|c_j - c_i\|$, and hence,
it remains to bound $\|c_j - c_i\|$ from below.
If we can prove $\alpha \geq \eps/n$,
then we have a potential drop of at least
$(2x-c_i-c_j) \cdot \alpha v \ge \alpha 2 \eps \geq 2\eps^2/n$ as claimed.

Let $H$ be the hyperplane bisecting the centers of
$\cluster i$ and $\cluster j$ in the previous iteration. While $H$
does not necessarily bisect $c_i$ and $c_j$, it divides the data
points belonging to $\cluster i$ and $\cluster j$ correctly.
In particular, this implies that $\|c_i - c_j\| \geq \dist(c_i,
H)+\dist(c_j, H)$.

Consider the at least $2d+1$ data points switching between $\cluster i$
and $\cluster j$. One of them must be at a distance of at least $\eps$
of $H$ since $\points$ is $\eps$-separated. Let us assume w.l.o.g.\
that this point switches to $\cluster i$. This yields $\dist(c_i, H)
\geq \eps/n$ since $\cluster i$ contains at most $n$ points. Thus, $\|c_i - c_j\| \geq \eps/n$,
which yields $\alpha \geq \eps/n$ as desired.
\end{proof}

Now set $\delta_i = n^{-16 - (16+i) \cdot \sqrt k } \cdot \sigma$
and $\eps_i = \sigma \cdot n^{-4-i\sqrt k}$.
Then the probability that the instance is not $\delta_i$-sparse
is bounded from above by
\[
  n^{16d\sqrt k+12 + 4d - 16d - (16+i) d\cdot \sqrt k} \leq
  n^{- i d \sqrt k}\:.
\]
The probability that the instance is not $\eps_i$-separated is
bounded from above by (we use $d \leq n$ and $4 \leq n$)
\[
  n^{4d-4d-i d\sqrt k} = n^{-id\sqrt k}\:.
\]
We abbreviate the fact that an instance is $\delta_i$-sparse and
$\eps_i$-separated by \emph{$i$-nice}. Now Lemmas~\ref{lem:deltasparsedrop}
and~\ref{lem:epssepdrop} immediately yield the following lemma.

\begin{lemma}
Assume that $\points$ is $i$-nice. Then the number of sequences of at most four
consecutive iterations, each of which with at most $\sqrt k$ active clusters,
until the potential has dropped by at least $1$ is bounded from above by
\[
  \left(\min\left\{
   \frac14\cdot n^{-36 - (32+2i) \sqrt k}\cdot
   \sigma^2,
   2\sigma^2 \cdot n^{-9-i2\sqrt k}
   \right\}\right)^{-1}
   \leq \frac{n^{(c+2i) \cdot \sqrt k}}{\sigma ^2} =: S_i
\]
for a suitable constant $c$.
\end{lemma}

The first term comes from $\delta_i$,
which yields a potential drop of at least $\delta_i^2/(4 n^4)$.
The second term comes from $\eps_i$, which yields a drop of at least
$2\eps_i^2/n$. 

Putting the pieces together yields the main lemma of this section.

\begin{lemma}
\label{lem:smallsqrtk}
The expected number of sequences of at most four consecutive iterations, each of
which with at most $\sqrt k$ active clusters, until the potential has dropped by
at least $1$ is bounded from above by
\[
  \poly\left(n^{\sqrt k}, \frac 1{\sigma}\right).
\]
\end{lemma}

\begin{proof}
The probability that it takes more than $S_i$ such sequences is bounded from
above by the probability that the instance is not $i$-nice, which is bounded
from above by $2  n^{- i d\cdot \sqrt k}$. Let $T$ be the random variable of the
number of sequences of at most four consecutive iterations, each with at most
$\sqrt k$ active clusters, that it takes until we have a potential drop of at
least $1$.

We observe that $k$-means runs always at most $W \leq n^{\kappa kd}$
iterations. This yields that we have to consider $i$ only up to $\kappa\sqrt k$. We assume
further that $d \geq 2$. Putting all observations together yields the lemma:
\begin{align*}
 \Ex T & \leq S_0+\sum_{i=0}^{\kappa \sqrt k} \Pr{\text{not $i$-nice, but
 $(i+1)$-nice}} \cdot S_{i+1}+W\cdot\Pr{\text{not $\kappa\sqrt{k}$-nice}}\\ 
  & \leq S_0+\sum_{i=0}^{\kappa \sqrt k} \Pr{\text{not $i$-nice}}
         \cdot S_{i+1}+n^{\kappa kd}\cdot2n^{-\kappa dk} \\
  &\leq \frac{n^{c \cdot \sqrt k}}{\sigma ^2}+ 
  \sum_{i=0}^{\kappa \sqrt k} 2  n^{- i d\cdot \sqrt k}
         \cdot \frac{n^{(c+2+2i) \cdot \sqrt k}}{\sigma^2}+2\\ 
  & \leq \frac{n^{c \cdot \sqrt k}}{\sigma ^2}+ 
  \sum_{i=0}^{\kappa \sqrt k} 2 \cdot \frac{n^{(c+2) \sqrt
  k}}{\sigma^2}+2 \leq \poly\left(n^{\sqrt k}, \frac 1{\sigma}\right) \: .
  \qedhere
\end{align*}
\end{proof}

\section[Iterations with at least sqrt(k) Active Clusters]
        {Iterations with at least $\sqrt{k}$ Active Clusters}
\label{sec:atleastsqrtk}

In this section, we consider steps of the $k$-means algorithm in which at least
$\sqrt{k}$ different clusters gain or lose points. The improvement yielded by
such a step can only be small if none of the cluster centers changes its
position significantly due to the reassignment of points, which, intuitively,
becomes increasingly unlikely the more clusters are active. We show that,
indeed, if at least $\sqrt{k}$ clusters are active, then with high probability
one of them changes its position by $n^{-O(\sqrt{k})}$, yielding a potential
drop in the same order of magnitude. 

The following observation, which has also been used by Arthur and
Vassilvitskii~\cite{ArthurVassilvitskii:ICP:2006}, relates the movement of a
cluster center to the potential drop.

\begin{lemma}
\label{lemma:MovementImprovement}
If in an iteration of the $k$-means algorithm a cluster center changes
its position from $c$ to $c'$, then the potential drops by at least
$\norm{c-c'}^2$.
\end{lemma}

\begin{proof}
The potential is defined as
\[
  \sum_{x\in \points}\norm{x-c_x}^2\:,
\]
where $c_x$ denotes the center that is closest to $x$. We can rewrite
this as
\[
  \sum_{c\in\clusterSet}\sum_{x\in X_c}\norm{x-c}^2
  = \sum_{c\in\clusterSet}\left(
       \sum_{x\in X_c}\norm{x-\mass(X_c)}^2+|X_c|\cdot\norm{\mass(X_c)-c}^2
       \right)\:,
\]
where $X_c\subseteq \points$ denotes all points from $\points$ whose
closest center is $c$ and where $\mass(X_c)$ denotes the center of mass of $X_c$.

Let us consider the case that one center changes its position from $c$
to $c'$. Then $c'$ must be the center of mass of $X_c$. Furthermore,
$|X_c| \geq 1$. Hence, the potential drops by at least
\[
    \norm{\mass(X_c)-c}^2-\norm{\mass(X_c)-c'}^2
  = \norm{c'-c}^2-\norm{c'-c'}^2
  = \norm{c'-c}^2\:.\qedhere
\]
\end{proof}

Now we are ready to prove the main lemma of this section.

\begin{lemma}
\label{lem:largesqrtk}
The expected number of steps with at least $\sqrt{k}$ active clusters
until the potential drops by at least $1$ is bounded from above by
\[
   \poly\left(n^{\sqrt k}, \frac 1{\sigma}\right)\enspace.
\]
\end{lemma}

\begin{proof}
We consider one step of the $k$-means algorithm with at least
$\sqrt{k}$ active clusters. Let $\eps$ be defined as in
Lemma~\ref{lemma:NumberOfClosePoints} for $a=1$. We distinguish
two cases: Either one point that is reassigned during the considered
iteration has a distance of at least $\eps$ from the bisector that it crosses,
or all points are at a distance of at most $\eps$ from their respective bisectors.
In the former case, we immediately get a potential drop of at least
$2\eps\Delta$, where $\Delta$ denotes the minimal distance of two
cluster centers. In the latter case,
Lemma~\ref{lemma:NumberOfClosePoints} implies that with high
probability less than $kd$ points are reassigned during the considered
step. We apply a union bound over the choices for these points.
In the union bound, we fix not only these points but also the clusters
they are assigned to before and after the step. We denote by $A_i$ the
set of points that are assigned to cluster $\cluster i$ in both
configurations and we denote by $B_i$ and $B_i'$ the sets of points assigned to
cluster $\cluster i$ before and after the step, respectively,
except for the points in $A_i$. Analogously to
Lemma~\ref{lemma:NumberOfClosePoints}, we assume that the positions of the points in
$A_1\cup\ldots\cup A_k$ are fixed adversarially, and we apply
a union bound on the different partitions $A_1,\ldots,A_k$ realizable.
Altogether, we have a union bound over less than
\[
  n^{\kappa kd}\cdot n^{3kd} \le n^{(\kappa+3) \cdot kd}
\]
events.
Let $c_i$ be the position of the cluster center of
$\cluster i$ before the reassignment,
and let $c_i'$ be the position after the reassignment. Then
\[
   c_i = \frac{|A_i|\cdot \mass(A_i)+|B_i|\cdot
   \mass(B_i)}{|A_i|+|B_i|}\enspace,
\]
where $\mass(\cdot)$ denotes the center of mass of a point set. Since
$c_i'$ can be expressed analogously, we can write the change of
position of the cluster center of $C_i$ as
\[
   c_i - c_i' = 
   |A_i|\cdot
   \mass(A_i)\left(\frac{1}{|A_i|+|B_i|}-\frac{1}{|A_i|+|B_i'|}\right)
   +\frac{|B_i|\cdot \mass(B_i)}{|A_i|+|B_i|} -\frac{|B_i'|\cdot
   \mass(B_i')}{|A_i|+|B_i'|}\enspace.
\]

Due to the union bound, $\mass(A_i)$ and $|A_i|$ are fixed. Additionally,
also the sets $B_i$ and $B_i'$ are fixed but not the positions of
the points in these two sets. If we considered only a single center, then
we could easily estimate the probability that
$\norm{c_i-c_i'}\le\beta$. For this, we additionally fix all positions
of the points in $B_i\cup B_i'$ except for one of them, say $b_i$.
Given this, we can express the event $\norm{c_i-c_i'}\le\beta$ as the event
that $b_i$ assumes a position in a ball whose position depends on the
fixed values and whose radius, which depends on the number of points in 
$|A_i|$, $|B_i|$, and $|B_i'|$, is not larger than $n\beta$. Hence,
the probability is bounded from above by
\[
  \left(\frac{n\beta}{\sigma}\right)^d\enspace.
\]

However, we are interested in the probability that this is true for all
centers simultaneously. Unfortunately, the events are not independent for
different clusters. We estimate this probability by identifying a set of $\ell/2$
clusters whose randomness is independent enough, where
$\ell\ge\sqrt{k}$ is the number of active clusters. To be more
precise, we do the following: Consider a graph whose nodes are the active
clusters and that contains an edge between two nodes if and only if the
corresponding clusters exchange at least one point. We identify a
dominating set in this graph, i.e., a subset of nodes that covers the
graph in the sense that every node not belonging to this subset has at
least one edge into the subset. We can assume that the dominating
set, which we identify, contains at most half of the active
clusters. (In order to find such a dominating set, start with
the graph and throw out edges until the remaining graph is a tree.
Then put the nodes on odd layers to the left side and the nodes on even layers
to the right side, and take the smaller side as the dominating set.)

For every active center $C$ that is not in the dominating set, we do the
following: We assume that all the positions of the points in $B_i\cup
B_i'$ are already fixed except for one of them. Given this, we can use
the aforementioned estimate for the probability of
$\norm{c_i-c_i'}\le\beta$. If we iterate this over all points not in the
dominating set, we can always use the same estimate; the reason is that
the choice of the subset guarantees that, for every node not in the subset,
we have a point whose position is not fixed yet. This yields an upper bound of
\[
  \left(\frac{n\beta}{\sigma}\right)^{d\ell/2}\enspace.
\]
Combining this probability with the number of choices in the union
bound yields a bound of 
\[
  n^{(\kappa +3) \cdot
  kd}\cdot\left(\frac{n\beta}{\sigma}\right)^{d\ell/2}
  \le
  n^{(\kappa +3) \cdot
  kd}\cdot\left(\frac{n\beta}{\sigma}\right)^{d\sqrt{k}/2}\enspace.
\]
For
\[
  \beta = \frac{\sigma}{n^{(4\kappa +6)\cdot \sqrt{k}+1}}
\]
the probability can be bounded from above by $n^{-\kappa kd}\le W^{-1}$.

Now we also take into account the failure probability of $2W^{-1}$ from
Lemma~\ref{lemma:NumberOfClosePoints}. This yields that, with a probability
of at least $1-3W^{-1}$, the potential drops in
every iteration, in which at least $\sqrt k$ clusters are active, by at least
\begin{align*}
   \Gamma & := \min\{2\eps\Delta,\beta^2\}
   \ge \min\left\{
   \frac{\sigma^{8}\Delta}{1296n^{14+8\kappa}D^{6}d},
   \frac{\sigma^2}{n^{(8 \kappa + 12)\cdot\sqrt{k}+2}}
   \right\} \\
& \geq \min\left\{\Delta \cdot \poly\left(n^{-1}, \sigma\right), \poly\left(n^{-\sqrt k}, \sigma
\right)\right\} 
\end{align*}
since $d \leq n$ and $D$ is polynomially bounded in $\sigma$ and $n$.
The number $T$ of steps with at least $\sqrt{k}$ active clusters until
the potential has dropped by one can only exceed $t$ if $\Gamma\le
1/t$. Hence,
\begin{align*}
   \Ex{T} & \le \sum_{t=1}^{\infty}\Pr{T\ge t} +3W^{-1}\cdot W
   \le 3+\int_{t=0}^{\infty}\Pr{T\ge t}\,\dint t \\
   & \le 4+\int_{t=1}^{\infty}\PrB{\Gamma\le\frac{1}{t}}\,\dint t
   \: \le
   4+\beta^{-2}+\int_{t=\beta^{-2}}^{\infty}\PrB{\Gamma\le\frac{1}{t}}\,\dint t\\ 
   & \le 4+\beta^{-2}+\int_{t=\beta^{-2}}^{\infty}
   \PrB{\Delta \cdot \poly\left(\frac 1n, \sigma\right)
   \le\frac{1}{t}}\,\dint t\\
   & \le 4+\beta^{-2}+\int_{t=\beta^{-2}}^{\infty}
   \PrB{\Delta\le 
   \frac 1t \cdot \poly\left(n, \frac 1\sigma\right)
}\,\dint t\\
   & \le 4+\beta^{-2}+\int_{t=\beta^{-2}}^{\infty}
   \min\left\{1, 
   \left(\frac{(4d+16) \cdot n^4 \cdot \poly\bigl(n, \sigma^{-1}\bigr)}{t \cdot \sigma}\right)^d
\right\}\,\dint t
   \: = \poly\left(n^{\sqrt k}, \frac 1\sigma\right)\enspace,
\end{align*}
where the integral is upper bounded as in the proof of Lemma~\ref{lem:GeneralDropBy1}.
\end{proof}

\section{A Polynomial Bound in One Dimension}
\label{sec:1d}

In this section, we consider a one-dimensional set $\points \subseteq \RR$
of points. The aim of this section is to prove that the expected number of steps
until the potential has dropped by at least $1$ is bounded by a polynomial in
$n$ and $1/\sigma$.

We say that the point set $\points$ is \emph{$\eps$-spreaded} if the following
conditions are fulfilled:
\begin{itemize}
\item There is no interval of length $\eps$ that contains three or more points
      of $\points$.
\item For any four points $x_1, x_2, x_3, x_4$, where $x_2$ and $x_3$ may denote the same point,
      we have $|x_1-x_2| > \eps$
      or $|x_3-x_4| > \eps$.
\end{itemize}
The following lemma justifies the notion of $\eps$-spreadedness.

\begin{lemma}
Assume that $\points$ is $\eps$-spreaded. Then the potential drops
by at least $\frac{\eps^2}{4n^2}$ in every iteration.
\end{lemma}

\begin{proof}
Let $\cluster i$ be the left-most active cluster, and let 
$\cluster j$ be the right-most active cluster.

We consider $\cluster i$ first. $\cluster i$ exchanges only
points with the clusters to its right, for otherwise it would not be the
leftmost active cluster. Thus, it cannot gain and lose points simultaneously.
Assume that it gains points. Let $A_i$ be the set of points of $\cluster i$
before the iteration, and let $B_i$ be the set of points that it gains.
Obviously, $\min B_i > \max A_i$.
If $B_i \cup A_i$ contains at least three points, then we are done:
If $|A_i|\ge2$, then we consider the two rightmost points $x_1\le x_2$
of $A_i$ and the leftmost point $x_3$ of $B_i$.
These points are not within a common interval of size $\eps$.
Hence, $x_3$ has a distance of at least $\eps/2$ from the center of mass
$\mass (A_i)$ because $\dist(x_1,x_3)\ge\eps$, $x_1\le x_2\le x_3$, and
$\mass (A_i)\le (x_1+x_2)/2$. Hence,
\[
  \mass(B_i) \geq \mass (A_i) + \frac{\eps}{2}.
\]
Thus, the cluster center moves to the right from $\mass(A_i)$ to
\[
  \mass(A_i \cup B_i)  =
  \frac{|A_i| \cdot \mass(A_i) + |B_i| \cdot \mass(B_i)}{|A_i \cup B_i|} 
  \geq \frac{|A_i \cup B_i| \cdot \mass(A_i)  +
  |B_i|\cdot \frac{\eps}2}{|A_i \cup B_i|} \geq \mass(A_i) +
  \frac{\eps}{2n}.
\]
The case $|A_i|=1$ and $|B_i|\ge 2$ is analogous. 
The same holds if cluster $\cluster j$ switches from $A_j$ to $A_j \cup B_j$
with $|A_j \cup B_j| \geq 3$, or if $\cluster i$ or $\cluster j$ lose
points but initially have at least three points. Thus, in these cases,
a cluster moves by at least $\eps/(2n)$, which causes a potential drop
by at least $\eps^2/(4n^2)$.

It remains to consider the case that $|A_i \cup B_i| = 2 = |A_j \cup B_j|$.
Thus, $A_i = \{a_i\}$, $B_i = \{b_i\}$, and also
$A_j = \{a_j\}$, $B_j = \{b_j\}$.
We restrict ourselves to the case that $\cluster i$ consists
only of $a_i$ and gains $b_i$ and that $\cluster j$ has $a_j$ and
$b_j$ and loses $b_j$.
If only two clusters are active, we have $b_i = b_j$, and we have only
three different points. Otherwise, all four points are distinct. This
allows us to bring $\eps$-spreadedness into play.
We have either $|a_i - b_i| \geq \eps$
or $|a_j - b_j| \geq \eps$. But then either the center of $\cluster i$
or the center of $\cluster j$ moves by at least $\eps/2$, which
implies that the potential decreases by at least $\eps^2/4 \geq
\eps^2/(4n^2)$.
\end{proof}

Assume that $\points$ is $\eps$-spreaded. Then the number of iterations
until the potential has dropped by at least $1$ is at most $4n^2/\eps^2$
by the lemma above. Let us estimate the probability that
$\points$ is $\eps$-spreaded.

\begin{lemma}
The probability that $\points$ is not $\eps$-spreaded is bounded from above
by $\frac{2n^4 \eps^2}{\sigma^2}$.
\end{lemma}

\begin{proof}
For the first property, let us consider any three points $x_1, x_2, x_3$,
and assume that $x_1$ is fixed
arbitrarily. Then, in order to share an interval of size $\eps$, we must
have $|x_i - x_1| \leq \eps$ for $i = 2,3$. Since $x_2$ and $x_3$ are
independent, this happens with a probability of at most
$\bigl(\frac{2\eps}{\sqrt{2\pi} \sigma}\bigr)^2 \leq \bigl(\frac \eps\sigma\bigr)^2$.
There are at most $n^3 \leq n^4$ choices for $x_1, x_2, x_3$.

For the second property, consider any $x_1, \ldots, x_4$, and assume that
$x_2$ and $x_3$ are fixed. Then the probability that $|x_1 - x_2| \leq \eps$
and $|x_3 - x_4| \leq \eps$ is at most $\bigl(\frac{\eps}{\sigma}\bigr)^2$.
There are at most $n^4$ choices for $x_1, \ldots, x_4$.

Overall, by a union bound, the probability that $\points$ is not $\eps$-spreaded
is at most $\frac{2n^4 \eps^2}{\sigma^2}$.
\end{proof}

Now we have all ingredients for the proof of the main lemma of this section.

\begin{lemma}
\label{lem:d1}
The number of iterations of $k$-means until the potential has dropped by at least $1$
is bounded by a polynomial in $n$ and $1/\sigma$.
\end{lemma}

\begin{proof}
Let $T$ be the random variable of the number of iterations until the potential has
dropped by at least $1$. If $T \geq t$, then $\points$ cannot be $\eps$-spreaded
with $4n^2/\eps^2 \leq t$. Thus, in this case, $\points$ is not
$\eps$-spreaded with $\eps = \frac{2n}{\sqrt t}$. In the worst case,
$k$-means runs for at most $n^{\kappa k}$ iterations. Hence,
\begin{align*}
\Ex T & = \sum_{t = 1}^{n^{\kappa k}} \Pr{T \geq t}
\leq \sum_{t = 1}^{n^{\kappa k}} \PrB{\text{$\points$ is not $\frac{2n}{\sqrt t}$-spreaded}} \\
& \leq \sum_{t = 1}^{n^{\kappa k}} \frac{8n^4 n^2}{t\sigma^2}
\in O\left(\frac{n^6}{\sigma^2} \cdot \log n^{\kappa k}\right) \subseteq
O\left(\frac{n^7}{\sigma^2} \cdot \log n\right)\:.\qedhere
\end{align*}
\end{proof}

Finally, we remark that, by choosing $\eps = \frac{\sigma}{n^{2+c}}$, we obtain that
the probability that the number of iterations until the potential has dropped by at least
exceeds a polynomial in $n$ and $1/\sigma$  is bounded from above by $O(n^{-2c})$.
This yields a bound on the running-time of $k$-means for $d = 1$ that holds with
high probability.

\section{Putting the Pieces Together}
\label{sec:putting}

In the previous sections, we have only analyzed the expected number of
iterations until the potential drops by at least $1$. To bound the expected
number of iterations that $k$-means needs to terminate, we need an upper
on the potential in the beginning. To get this, we use the following lemma.

\begin{lemma}
\label{lem:deviation}
Let $x$ be a one-dimensional Gaussian random variable with standard
deviation $\sigma$ and mean $\mu\in[0,1]$. Then,
for all $t \geq 1$,
\[
  \Pr{x\notin[-t,1+t]} < \sigma \cdot
  \exp\left(-\frac{t^2}{2\sigma^2}\right).
\]
\end{lemma}

For $D = \sqrt{2\sigma^2\ln(n^{1+\kappa kd}d\sigma)}\le\poly(n,\sigma)$,
the probability that any component of any
of the $n$ data points is not contained in the hypercube
$\cube=[-D,1+D]^d$ is bounded from above by $n^{-\kappa kd}\leq
W^{-1}$. This implies that $\points \subseteq \cube$ 
with a probability of at least $1-W^{-1}$.

In the beginning, we made the assumption that $\sigma \leq 1$. While this covers
the small values of $\sigma$, which we consider as more relevant, the assumption
is only a technical requirement, and we can get rid of it: The number
of iterations that $k$-means needs is invariant under scaling of the point set
$\points$. Now assume that $\sigma > 1$. Then we consider $\points$ scaled down
by $1/\sigma$, which corresponds to the following model: The adversary chooses
points from the hypercube $[0,1/\sigma]^d \subseteq [0,1]^d$, and then we add
$d$-dimensional Gaussian vectors with standard deviation $1$ to every
data point. The expected running-time that $k$-means needs on this instance
is bounded from above by the running-time needed for adversarial points chosen from
$[0,1]^d$ and $\sigma = 1$, which is $\poly(n) \leq \poly(n, 1/\sigma)$.

\subsection{Proof of Theorem~\ref{thm:d1}}

We obtain a bound that is polynomial in $n$ and $1/\sigma$
from Lemmas~\ref{lem:d1} and~\ref{lem:deviation}: First, after one
iteration, the potential is bounded from above by $\poly(n, D) = \poly(n)$.
If this is not the case, we bound the number of iterations by $W$, which
adds $W \cdot W^{-1}$ to the expected number of iterations.
Second, the expected number of iterations until the potential has dropped
by at last $1$ is bounded by $\poly(n, 1/\sigma)$, which yields
a bound of $\poly(n, 1/\sigma)$ until the algorithm terminates.
This proves Theorem~\ref{thm:d1}.

The result that the probability that the number of iterations
exceeds a polynomial in $n$ and $1/\sigma$ is at most $O(1/\poly(n))$
follows immediately.

\subsection{Proof of Theorem~\ref{thm:sqrtk}}
\label{ssec:sqrtk}

In the remainder of this section, we restrict ourselves to $d \geq 2$.
For $d = 1$, we already have a polynomial bound according to Theorem~\ref{thm:d1}
and Section~\ref{sec:1d}.

After $\poly\bigl(n^{\sqrt k}, 1/\sigma\bigr)$ iterations, we have
\begin{itemize}
\item at least $\poly\bigl(n^{\sqrt k}, 1/\sigma\bigr)$ sequences of four
      consecutive iterations, each of which with at most $\sqrt k$ active
      clusters, or
\item at least $\poly\bigl(n^{\sqrt k}, 1/\sigma\bigr)$ iterations with at least
      $\sqrt k$ active clusters.
\end{itemize}
Thus, by Lemmas~\ref{lem:smallsqrtk} and~\ref{lem:largesqrtk}, the expected
number of steps until the potential has dropped by one is at most
$\poly\bigl(n^{\sqrt k}, 1/\sigma\bigr)$.

After the first iteration, the potential is at most $nd \cdot (2D+1)^2$. As
argued above, we can restrict ourselves to $\sigma \leq 1$, which implies
$D \leq \poly(n)$. This yields that the expected number of steps until
termination is at most
\[
    nd \cdot (2D+1)^2 \cdot \poly\left(n^{\sqrt k},\frac{1}\sigma\right)
  = \poly\left(n^{\sqrt{k}}, \frac 1{\sigma}\right) \: ,
\]
provided that $\points \subseteq \cube$. If $\points \not\subseteq \cube$, we
bound the number of iterations by the worst-case bound of $W$, which contributes
only $W^{-1} \cdot W = 1$ to the expected number of iterations. This proves
Theorem~\ref{thm:sqrtk}.

\subsection{Proofs of Theorem~\ref{thm:o1} and Corollary~\ref{cor:poly}}
\label{ssec:o1}

We exploit Lemma~\ref{lem:GeneralDropBy1} with $a = 2$.
Then the expected number of iterations until the potential has dropped
by at least $1$ is bounded from above by
\[
 k^{kd} \cdot \poly\left(n, \frac 1\sigma \right) \: .
\]
Again, after the first iteration, the potential is at most
$nd \cdot (2D+1)^2 \leq \poly(n)$ with a probability of at least $1-W^{-1}$.
This shows that the expected number of iterations, provided $\points \subseteq \cube$,
is bounded from above by
\[
k^{kd} \cdot \poly\left(n, \frac 1\sigma\right)\:.
\]
The event $\points\not\subseteq \cube$ contributes only $1$ to the expected number
of iterations as argued in the previous section.
This completes the proof of Theorem~\ref{thm:o1}.

If $k, d \in O\bigl(\sqrt{\log n/\log \log n}\bigr)$, then
$k^{kd} \leq \poly(n)$, which proves Corollary~\ref{cor:poly}.

\section{Conclusions}

We have proved two upper bounds for the smoothed running-time of the $k$-means
method: The first bound is $\poly(n^{\sqrt k}, 1/\sigma)$. The second bound is
$k^{kd} \cdot \poly(n, 1/\sigma)$, which decouples the exponential growth in $k$
and $d$ from the number of points and the standard deviation. In particular,
this yields a smoothed running-time that is polynomial in $n$ and $1/\sigma$
for $k, d \in O(\sqrt{\log n/\log \log n})$.

The obvious question now is whether a bound exists that is polynomial in $n$ and
$1/\sigma$, without exponential dependence on $k$ or $d$.
We believe that such a bound exists. However, we suspect that
new techniques are required to prove it; bounding the smallest possible improvement
from below might not be sufficient. The reason for this is that the number of
possible partitions, and thus the number of possible $k$-means steps, grows
exponentially in $k$, which makes it more likely for small improvements to exist
as $k$ grows.

Finally, we are curious if our techniques carry over to other heuristics. In
particular Lemma~\ref{lemma:NumberOfClosePoints} is quite general, as it bounds the
number of points from above that are close to the boundaries of the Voronoi partitions
that arise during the execution of $k$-means. In fact, we believe that
a slightly weaker version of Lemma~\ref{lemma:NumberOfClosePoints} is
also true for arbitrary Voronoi partitions and not only for those
arising during the execution of $k$-means. This insight might turn out to be
helpful in other contexts as well.

\section*{Acknowledgement}

We thank David Arthur, Dan Spielman, Shang-Hua Teng, and Sergei Vassilvitskii
for fruitful discussions and comments.


\begin{thebibliography}{10}

\bibitem{Arthur:Comm:2008}
David Arthur.
\newblock Smoothed analysis of the $k$-means method.
\newblock Manuscript, 2008.

\bibitem{ArthurVassilvitskii:HowSlow:2006}
David Arthur and Sergei Vassilvitskii.
\newblock How slow is the $k$-means method?
\newblock In Nina Amenta and Otfried Cheong, editors, {\em Proc. of the 22nd
  ACM Symposium on Computational Geometry (SOCG)}, pages 144--153. ACM Press,
  2006.

\bibitem{ArthurVassilvitskii:ICP:2006}
David Arthur and Sergei Vassilvitskii.
\newblock Worst-case and smoothed analysis of the {ICP} algorithm, with an
  application to the $k$-means method.
\newblock In {\em Proc. of the 47th Ann. IEEE Symp. on Foundations of Computer
  Science (FOCS)}, pages 153--164. IEEE Computer Society, 2006.

\bibitem{Badoiu}
Mihai B{\u{a}}doiu, Sariel Har-Peled, and Piotr Indyk.
\newblock Approximate clustering via core-sets.
\newblock In {\em Proc. of the 34th Ann.\ ACM Symposium on Theory of Computing
  (STOC)}, pages 250--257. ACM Press, 2002.

\bibitem{Berkhin}
Pavel Berkhin.
\newblock Survey of clustering data mining techniques.
\newblock Technical report, Accrue Software, San Jose, CA, USA, 2002.

\bibitem{Duda}
Richard~O. Duda, Peter~E. Hart, and David~G. Stork.
\newblock {\em Pattern Classification}.
\newblock John Wiley \& Sons, 2000.

\bibitem{HarPeled}
Sariel Har-Peled and Bardia Sadri.
\newblock How fast is the $k$-means method?
\newblock {\em Algorithmica}, 41(3):185--202, 2005.

\bibitem{InabaEA:WeightedVoronoi:2000}
Mary Inaba, Naoki Katoh, and Hiroshi Imai.
\newblock Variance-based $k$-clustering algorithms by {V}oronoi diagrams and
  randomization.
\newblock {\em IEICE Transactions on Information and Systems},
  E83-D(6):1199--1206, 2000.

\bibitem{Kumar}
Amit Kumar, Yogish Sabharwal, and Sandeep Sen.
\newblock A simple linear time $(1+\varepsilon$)-approximation algorithm for
  $k$-means clustering in any dimensions.
\newblock In {\em Proc. of the 45th Ann. IEEE Symp. on Foundations of Computer
  Science (FOCS)}, pages 454--462, 2004.

\bibitem{Lloyd}
Stuart~P. Lloyd.
\newblock Least squares quantization in {PCM}.
\newblock {\em IEEE Transactions on Information Theory}, 28(2):129--137, 1982.

\bibitem{Matousek}
Ji{\v{r}}{\'i} Matou{\v{s}}ek.
\newblock On approximate geometric $k$-clustering.
\newblock {\em Discrete and Computational Geometry}, 24(1):61--84, 2000.

\bibitem{SpielmanTeng:SmoothedAnalysisWhy:2004}
Daniel~A. Spielman and Shang-Hua Teng.
\newblock Smoothed analysis of algorithms: Why the simplex algorithm usually
  takes polynomial time.
\newblock {\em Journal of the ACM}, 51(3):385--463, 2004.

\bibitem{SpielmanTeng:SmoothedFoCM:2006}
Daniel~A. Spielman and Shang-Hua Teng.
\newblock Smoothed analysis of algorithms and heuristics: Progress and open
  questions.
\newblock In Luis~M. Pardo, Allan Pinkus, Endre S{\"u}li, and Michael~J. Todd,
  editors, {\em Foundations of Computational Mathematics, Santander 2005},
  pages 274--342. Cambridge University Press, 2006.

\end{thebibliography}

%

\end{document}